\documentclass[11pt,letterpaper]{article}
\usepackage{amsmath}
\usepackage{amsfonts}
\usepackage{amssymb}
\usepackage{graphicx}
\usepackage{amsmath, amssymb, amsthm}
\usepackage[left=1.00in, right=1.00in, top=1.00in, bottom=1.00in]{geometry}

\begin{document}

\newcommand{\reals}{\mathbb{R}}
\newtheorem{theorem}{Theorem}[section]
\newtheorem{corollary}[theorem]{Corollary}
\newtheorem{lemma}[theorem]{Lemma}
\newtheorem{proposition}[theorem]{Proposition}
\newtheorem{claim}[theorem]{Claim}
\newtheorem{definition}[theorem]{Definition}
%

\title{Reallocation Mechanisms}

\author{Liad Blumrosen \and Shahar Dobzinski}

\maketitle

\begin{abstract}
We consider reallocation problems in settings where the initial endowment of each agent consists of a subset of the resources. The private information of the players is their value for every possible subset of the resources. The goal is to redistribute resources among agents to maximize efficiency. Monetary transfers are allowed, but participation is voluntary.

We develop incentive-compatible, individually-rational and budget balanced mechanisms for several classic settings, including bilateral trade, partnership dissolving, Arrow-Debreu markets, and combinatorial exchanges. All our mechanisms (except one) provide a constant approximation to the optimal efficiency in these settings, even in ones where the preferences of the agents are complex multi-parameter functions.
\end{abstract}

\thispagestyle{empty}\newpage\setcounter{page}{1}

\section{Introduction}

One of the most fundamental problem in economics is to determine how to allocate scarce resources. Initially, resources may be inefficiently distributed among agents. However, as agents value resources differently, they may want to trade to improve their well being.

When every agent seeks to maximize his own utility, classic economic theory generally predicts the existence of an ``invisible hand'': agents will trade among themselves to maximize their own utility, and will eventually arrive at an efficient resource allocation. This paradigm fails, however, in the presence of asymmetric information. Private information may lead to market failures, where trade fails to take place even when it is desirable for all. One influential example is Akerlof's market for used cars \cite{Ake70}, where the unraveling of markets leads to no trade at all. In this paper we aim to design mechanisms that will foster trade in such markets, even if resources are distributed among multiple agents with different interests and partial information.

In exchange economies, where each agent is initially endowed with some resources, agents simultaneously play the role of buyers and sellers. Exchange economies are related to many real life environments; Individuals hold assets, like real estate, cars and stocks,
for which other individuals have their own preferences.
Firms hold other types of assets (e.g., employees, land, machines, intellectual property) that may possibly be better assigned if more information becomes available.
Numerous examples from the realms of industrial organization and finance fall into this model, like the dissolving of partnerships, breaking monopolies and other anti-trust acts, and merges and acquisitions. Of particular interest are structured markets where trade can be coordinated by centralized mechanisms.
One recent example for a centralized large-scale reallocation mechanism is the FCC's attempt to reallocate frequencies currently held by TV broadcasters to wireless phone companies (see \cite{SM13}). A major challenge in these FCC two-sided auctions is to provide incentives for the TV broadcasters to relinquish their licenses (see also \cite{ALMS13}).

Our main goal is to design markets that allow an efficient reallocation of resources. Technically, this translates to three requirements. The first one is \emph{individually rationality}: the participation of the agents is voluntary and at any point they may leave the market and consume their initial endowments. Thus, the outcome of any individually rational mechanism is a Pareto improvement in the economy, where agents are expected to be (weakly) better off in the new allocations. The second requirement is \emph{budget balance}: the mechanism is not allowed to subsidize the agents in order to improve the outcome. We distinguish between weakly and \emph{strongly budget balanced} mechanisms: in the latter the mechanism is additionally not allowed to burn money\footnote{The budget balance requirement is common in the cost-sharing literature (e.g., \cite{M99} and \cite{RS06}) but there the idea is to charge the participants an amount that suffices to cover the cost of providing the service.}.

The third requirement is \emph{truthfulness}. We discuss both Bayesian and prior-free models, but all our mechanisms admit ex-post dominant strategy equilibria (i.e., \emph{universal} truthfulness). Even when distributional assumptions are made, we make minimal use of this knowledge, namely we only require access to  statistical properties like the edowment's median value\footnote{In fact, using noisy estimations of the medians decreases the performance of our mechanism in a rate proportional to the noise. Hence, even if we only have a black box access to the distributions, we can use the black box to estimate the medians within an arbitrary precision, and preserve very similar performance guarantees.}.


Our reallocation problems are essentially combinatorial auctions where items are initially held by the players (and not by the auctioneer as usual), hence generalizing models of double auctions double auctions (see, e.g., \cite{RSW94,SW02,CS06,FMS07, DRT14}). This adds another layer of complexity; For example, while VCG can always be used to maximize the welfare in combinatorial auctions, in the presence of endowments no truthful mechanism can allocate efficiently and remain budget balanced, 
as we will shortly see.

The challenges one faces in designing reallocation mechanisms reveal themselves even when analyzing the simplest environment considered in this paper, \emph{bilateral trade}. We then build upon the insights we gain from this simple setting and develop mechanisms for more complex environments.

\subsubsection*{Bilateral Trade and Partnership Dissolving}


Consider a seller that holds a single indivisible good and a buyer. The seller's value for the good $v_s$ is drawn from some distribution $D_s$, and the buyer's value $v_b$ is independently drawn from $D_b$. In efficient markets the seller will sell the item if and only if $v_s< v_b$. Any price between these values will clear the market and support an efficient allocation. However, finding this price may be challenging, as the parties will try to influence the sale price by their bids.

A seminal impossibility result by Myerson and Satterthwaite \cite{MS83} shows that no mechanism can simultaneously achieve full efficiency, individual rationality and budget balance in (Bayes Nash) equilibrium. Our first result bounds the loss of efficiency in bilateral trade: it considers the \emph{Median Mechanism} that simply sets a trade price $p$ that equals to the median of the distribution of the seller; The seller sells the item at price $p$ to the buyer if and only if $v_s<p\leq v_b$. Notice that this mechanism is truthful, (strongly) budget balanced, and individually rational.

\vspace{0.1in} \noindent \textbf{Theorem:} The expected welfare of the median mechanism is at least half of the expected value of the efficient allocation. I.e., the median mechanism is a $2$ approximation to the optimal welfare.

\vspace{0.1in} \noindent We note that a similar mechanism that posts the median of the buyer's value does not provide any bounded approximation. However, by carefully analyzing both distributions we show how to find a trade price that gives an approximation ratio strictly better than $2$. The mechanism is inspired by a neat recent work by McAfee \cite{mcA08} who showed that posting any price between the medians of the two agents achieves in expectation a $2$ approximation to the gain from trade (defined to be $\max(0,v_b-v_s)$), but only if the median of the buyer is at least the median of the seller.

Notice that our efficiency benchmark in this paper is the value of the efficient allocation, i.e., the allocation that maximizes the social welfare, sometimes also called the \emph{first-best} solution. This allocation is clearly infeasible when players behave strategically \cite{MS83}.\footnote{An alternative approach -- when the distribution of preferences is known to the designer -- is to run the \emph{second-best} mechanism, i.e., \emph{the} mechanism that maximizes efficiency subject to the individual rationality and budget balance constraints. While this approach might be feasible in some very simple settings where the structure of the second-best mechanism is well understood (like bilateral trade with monotone-hazard rate distributions), very little is known on the structure of such mechanisms in the more complicated models we study later. This holds in particular for the multi-parameter environments we consider, for which a characterization of optimal mechanisms is a big open question.}

Next, we consider the more general setting of \emph{Partnership Dissolving}. Here, fractions of a divisible item are held by $n$ agents who have different (linear) preferences for this item. The goal is to reallocate the items in a way that maximizes the welfare, which, in this setting, is equivalent to giving all the good to the player with the highest value. Cramton, Gibbons and Klemperer \cite{CGK87} show that if the shares are close enough to equal holdings, there exists a fully efficient mechanism (see also \cite{mcA92}, and a survey \cite{Mol01} ). Their mechanism is Bayes Nash incentive compatible and interim individually rational. In contrast, we show that:

\vspace{0.1in} \noindent \textbf{Theorem:}
\begin{itemize}
\item Suppose that the shares are close enough to equal holdings. There exists a dominant strategy, ex-post individually rational, strongly budget balanced mechanism whose efficiency loss approaches $0$ as the number of players grows. This mechanism does not require any distributional knowledge on the values of the agents.

\item For any arbitrary initial division of shares, there exists a dominant strategy, ex-post individually rational, strongly budget balanced mechanism that provides a better-than-$2$ approximation to the optimal social welfare.
\end{itemize}

\vspace{0.1in} \noindent The second mechanism is obtained by introducing a general reduction: any $\alpha$-approximation mechanism for bilateral trade yields an $\alpha$-approximation for partnership dissolving, regardless of the initial division of shares.

Both bilateral trade and partnership dissolving are single-parameter domains. We now move on to study more complex multi-parameter domains. Indeed, our most technical constructions are developed for the next two multi-parameter environments.

\subsubsection*{Combinatorial Exchanges}

We first consider \emph{combinatorial exchanges}: a set of indivisible heterogeneous items that is initially distributed among the agents so that each agent $i$ holds a (possibly empty) subset $E_i$ of the items. Our players have combinatorial multi-parameter preferences over sets of items. Let $H_n$ denote the $n$'th harmonic number, and let $t=\max_i|E_i|$ (i.e., the maximal number of items held by a single player). We are able to show that:

\vspace{0.1in} \noindent \textbf{Theorem:} There exists a truthful, individually rational, weakly budget balanced, randomized mechanism that provides an $8H_{t}$-approximation to the optimal welfare if all valuations are subadditive\footnote{A valuation $v$ is subadditive if for ever two bundles $S$ and $T$ we have that $v(S)+v(T)\geq v(S\cup T)$.}. The only distributional knowledge that the mechanism uses is the median value of the distribution of the endowment of each player.

\vspace{0.1in} \noindent In particular, if each bidder is initially endowed with at most one item we get an $8$-approximation.

To gain some intuition about the mechanism, let us consider the simpler setting in which one bidder $i$ initially holds all items (that is, $E_i=M$). Let $MED_i$ denote the median of the value of the distribution of $v_i(E_i)$. We could have hoped to have the following generalization of the median mechanism for bilateral trade: bidder $i$ will report us if he is ready to sell all items for a price of $MED_i$. If he agrees, we will use VCG to find an optimal allocation of the items to all bidders but bidder $i$, and use the revenue generated from VCG to pay an amount of $MED_i$ to bidder $i$. If bidder $i$ does not agree, he keeps the items. Notice that the approximation ratio of this (incorrect) procedure is constant: if most of the expected optimal welfare is contributed by bidder $i$ then by doing nothing we already get a $2$ approximation, and the outcome of any valid mechanism is a Pareto improvement. On the other hand, if most of the expected optimal welfare is contributed by the rest of the bidders, then we get a $4$ approximation: with probability $\frac 1 2$ bidder $i$ sells his endowment, and in that case we allocate the items optimally among all bidders but bidder $i$.

Of course, the procedure above fails because we cannot guarantee that the revenue will from the VCG will be at least $MED_i$. To handle this, we develop a ``revenue extracting'' procedure which is the combinatorial auctions analogue of a second-price auction with reserve price. In a second-price auction the auctioneer can put a reserve price $r$ to guarantee revenue of at least $r$ if the highest value is at least $r$. We show that in a combinatorial auction with $n$ players there exists a (deterministic, prior-free) mechanism that guarantees a revenue of at least $r$ if the optimal welfare is at least $H_n\cdot r$.\footnote{If the optimal welfare is smaller than $H_n\cdot r$, then the mechanism is not required to allocate the items, but if it does so the revenue is guaranteed to be at least $r$.}

Our mechanism (for the special case) now looks as follows: allow bidder $i$ to sell all his items at price $MED_i$. If bidder $i$ agrees, we use the revenue extraction procedure with a ``reserve price'' of $r=MED_i$ to distribute the items among all bidders but bidder $i$, and use the revenue to pay $MED_i$ to bidder $i$. The mechanism for the general case can be found in Section \ref{sec-combinatorial}.

\subsubsection*{Prior-Free Mechanisms: Arrow-Debreu Markets}

Our second main technical construction considers the classic exchange model of Arrow and Debreu \cite{AD54}. We have a single divisible good, but the valuations are no longer linear as in partnership dissolving, but can be any function with decreasing marginals (i.e., concave valuations). An easy adaption of the mechanism for combinatorial exchanges guarantees a constant approximation, but our challenge now is to get rid of the distributional assumptions and develop \emph{prior free} mechanisms with a constant approximation ratio in the worst case.

From a technical perspective this is a multi-parameter environment for which our machinery for developing truthful mechanisms (especially prior free ones) is limited. Yet, to our surprise we were able to utilize ideas from the mechanism for combinatorial exchanges and come up with a constant-approximation mechanism that does not make any distributional assumptions.

\vspace{0.1in} \noindent \textbf{Theorem:} There exists a truthful, prior-free, individually rational, weakly budget balanced, randomized mechanism that provides a constant approximation to the optimal social welfare, as long as no player is initially endowed with more than $\frac 1 3$ of the good.

\vspace{0.1in} \noindent Notice the necessity of the last condition: in markets when, say, one player initially holds all the good, in the spirit of \cite{MS83} no prior-free mechanism with a constant approximation ratio exists.

A key idea in the mechanism is to replace the revenue extraction procedure that was used in the mechanism for combinatorial exchanges with a more subtle one, that will allow us to take advantage of the specifics of the setting and get a constant approximation ratio. Consider the special case where one agent $i$ holds all the good. Now, since we assume no distributional knowledge, we do not know the median value of the endowment of $i$, but let us assume for now that we know instead the ``mid-supply'' price: the price per fraction $p$ for which $i$ prefers to sell exactly half of his endowment. The crux is that since the valuation of $i$ exhibits decreasing marginals, agent $i$ will agree to sell any amount smaller than half of his endowment at the price. Thus, if we knew that mid-supply price we could just run VCG with the rest of the agents as well as an additional dummy bidder that has a value per fraction of $p$ for any amount below half of the endowment of $i$. Observe that if some bidder $i'$ was allocated fraction $x$ of the good, the VCG payment formula implies that his payment his at least $x\cdot p$ (since otherwise the dummy bidder can get an additional amount $x$ of the good). We can now take an amount of $x$ from agent $i$, assign it to $i'$ and pay agent $i$ $x\cdot p$.

\subsubsection*{Conclusions and Future Directions}

In this paper we devise welfare maximizing reallocation mechanisms. Almost all of our mechanisms provide a constant approximation ratio to the welfare-maximizing allocation, but we do not know whether these constants are optimal. In particular, proving impossibilities on the power of truthful and budget balanced mechanisms for reallocation problems is an interesting open question.

Our focus in this paper was not computational complexity, but it turns out that all of our mechanisms run in polynomial time, except the mechanism for combinatorial exchanges (see \cite{KSS04,BBK08} for computational issues in combinatorial exchanges).
Developing a polynomial time mechanism for the latter setting seems hard as in particular it implies a solution to the notorious problem of developing truthful polynomial time algorithm for combinatorial auctions with subadditive (and submodular) bidders (see, e.g., \cite{Dob11}, \cite{DV11}, and \cite{DV12}).

Finally, there are environments in which even the individual-rationality requirement can be relaxed. For example, regulators may force firms to participate in some markets even if it may hurt them (breaking monopolies and other anti-trust procedures fall into this category). Modeling such environments is an interesting question that may possibly lead to some practical insights.

\subsubsection*{Organization}

Section \ref{sec:model} defines a general framework which captures all settings we discuss. Instantiations of this framework are studied next: bilateral trade and partnership dissolving in Section \ref{sec-bilateral}, combinatorial exchanges in Section \ref{sec-combinatorial} and Arrow-Debreu markets in Section \ref{sec-multiunit}.

\section{The Framework}
\label{sec:model}

Consider a set of resources $M=\{1,...,m\}$ and a set of $n$ agents.
Let $\mathcal{E}_i \subseteq [0,1]^m$ be the set of allowed endowments for agent $i$.
Let $\mathcal{A} \subseteq [0,1]^{n \times m}$ be the set of allowed allocations of resources among the agents, where $\mathcal{A}_i \subseteq [0,1]^m$ stands for the set of possible allocations to player $i$.

The valuation of player $i$ is a function $v_i:\mathcal{A}_i \cup \mathcal{E}_i\rightarrow \reals_+$. Let $V_i$ be the set of all possible valuations of player $i$, and $V=V_1 \times ... \times V_n$.
We sometimes assume a Bayesian model, where $v_i$ is drawn from $V_i$ according to a distribution $F_i$, independently from the valuations of the other agents. The valuations are private information and the endowments are known to the designer.

A (direct revelation) reallocation mechanism consists of an allocation function $M:V \rightarrow \Delta(\mathcal{A})$ and a payment function $p:V \rightarrow \reals^n$.\footnote{Note that as agents in our model can be sellers and buyers simultaneously, we do not assume that payments are positive; Negative payments mean transfers from the mechanism to the agents.
}
All of our mechanisms are dominant-strategy truthful. That is, reporting the true valuations $v_i$ is a dominant strategy for every agent $i$. Truthful behavior is ex-post (rather than dominant-strategy in expectation) which allows us, e.g., to ignore distributional beliefs of the agents and whether they are risk-neutral or not. We require the following:
\begin{itemize}
\item \textbf{Ex Post Individual Rationality.} For every allocation and payment $A_i,p_i$ eventually allocated to agent $i$ with initial endowment $E_i \in \mathcal{E}_i$ (after the realization of the valuations and the randomness of the mechanism), we have that
    $v_i(A_i)- p_i \geq v_i(E_i)$.

\item \textbf{Ex Post Budget Balance.} For every $\mathbf{v}\in V$ of the preferences we have  $\sum_{i=1}^n p_i(\mathbf{v}) = 0$.
    If, instead, we only have that $\sum_{i=1}^n p_i(\mathbf{v}) \geq 0$, the mechanism is \textit{weakly budget balanced}.\footnote{
    Note that this definition holds for every realization of $v$ (and not only in expectation, which is usually a key for achieving budget balance in Bayesian domains, e.g., in \cite{CGK87}).
    }

\item \textbf{Approximate Efficiency}. We would like to approximate the optimal expected efficiency with non-strategic agents, $OPT=\max_{A \in \mathcal{A}} E_{\mathbf{v}\in V} \big[ \sum_{i=1}^n v_i(A_i)\big]$. A mechanism achieves an $\alpha$ approximation to the optimal welfare if $E[\sum_{i=1}^n v_i( M(\mathbf{v}) )] \geq \frac{OPT}{\alpha}$ (expectation is over the random coins of the mechanism, if any, and over the valuations $\mathbf{v}\in V$).
\end{itemize}

\section{Partnership Dissolving and Bilateral Trade}\label{sec-bilateral}

Our first set of results concerns bilateral trade and partnership dissolving. As we will see these two settings are closely related, in the sense that an approximation algorithm for the bilateral trade problem immediately implies an algorithm with the same approximation guarantee for partnership dissolving. We later extend the mechanisms we develop to more complex multi-parameter settings.

\subsection{Bilateral Trade}

In the bilateral trade problem a seller holds a single indivisible good. The seller's value for the good $v_s$ is drawn from some distribution $D_s$. There is also a buyer, whose value for the good $v_b$ is independently drawn from $D_b$. Let $M_s$ denote the median of the seller's distribution. The \emph{Median Mechanism} works as follows: if both players accept the price $M_s$ (i.e., $v_b \geq M_s$ and $v_s \leq M_s$) then the buyer gets the item and pays $M_s$ to the seller. Otherwise, the seller keeps the item and no payments are made.

\begin{theorem}
The \textit{Median Mechanism} is truthful, individually rational, budget balanced and achieves a $2$-approximation to the optimal social welfare.
\end{theorem}
\begin{proof}
The mechanism is obviously truthful, individually rational, and budget balanced. We now analyze its approximation ratio. We start with some notation. An instance $(v_s,v_b)$ is \textit{type 1}, if $v_b<M_s$. All other instances are \textit{type $2$}. Let $S_i$ denote the contribution to the optimal welfare of type $i$ instances in which the seller has a higher value, i.e., for $i \in \{ 1,2\} )$:
\begin{align*}
S_i=E[v_s|v_s\geq v_b, \hbox{ instance is type i}] \cdot \Pr[v_s\geq v_b, \hbox{ instance is type i}]
\end{align*}
Similarly, let $B_i=E[v_b|v_b > v_s, \hbox{ instance is type i}] \cdot \Pr[v_b> v_s, \hbox{ instance is type i}]$.

Define $OPT_i=S_i+B_i$ (for $i=1,2$) and observe that $OPT=OPT_1+OPT_2$. Let $ALG$ be the expected welfare achieved by the median mechanism.  Let $ALG_i$ be the expected contribution of type $i$ instances to this welfare ($ALG=ALG_1+ALG_2$). We will show that $ALG_1 \geq \frac{OPT_1}{2}$ and $ALG_2 \geq \frac{OPT_2}{2}$ and the theorem will follow.

\begin{claim}
$ALG_1 \geq \frac{OPT_1}{2}$.
\end{claim}
\begin{proof}
First, we claim that $B_1\leq \frac{M_S}{2}$. Indeed, for type $1$ instances
$v_b<M_s$. Hence,
\begin{align*}
B_1&=E[v_b|v_b > v_s, \hbox{ instance is type 1}] \cdot \Pr[v_b > v_s, \hbox{ instance is type 1}]< M_s \cdot \Pr[v_s \leq M_s]  = \frac{M_s}{2}
\end{align*}

Thus, $OPT_1=S_1+B_1\leq S_1+\frac{M_S}{2}$.
Now observe that $ALG_1 \geq S_1$ since the
mechanism never sells the item if $v_s\geq v_b$.
Finally, we claim that $ALG_1>M_S/2$, since 
given that the instances are of type 1, with probability $\frac 1 2$
 we have that $v_s>M_s$ ($v_s$ and $v_b$ are independent)
 and the welfare is at least $M_s$ (as $v_s>v_b$). Together we have that $ALG_1>OPT_1/2$.
\end{proof}

\begin{claim}
$ALG_2 \geq \frac{OPT_2}{2}$.
\end{claim}
\begin{proof}
Observe that in the median mechanism, the buyer buys the item in type $2$ instances only when the seller's value is below the price $M_s$. Thus,
\begin{align*}
ALG_2=S_2+E[v_b|\hbox{instance is type 2}, v_s<M_s] \cdot \Pr[\hbox{instance is type 2},v_s<M_s] .
 \end{align*}

Since $v_s$ and $v_b$ are drawn from independent distributions we have that
\begin{align}
ALG_2& =S_2+ E[v_b|\hbox{instance is type 2}] \cdot \Pr[\hbox{instance is type 2},v_s<M_s]\nonumber\\
& =S_2+ E[v_b|\hbox{instance is type 2}] \cdot \Pr[\hbox{instance is type 2}]\cdot \Pr[v_s<M_s]\nonumber\\
&= S_2+ E[v_b|\hbox{instance is type 2}] \cdot \Pr[\hbox{instance is type 2}]\cdot \frac{1}{2} \nonumber\\
&\geq \frac{1}{2} \big( S_2+ E[v_b|\hbox{instance is type 2}] \cdot \Pr[\hbox{instance is type 2}] \big) \label{eq:bilateral-2-approx-1}
\end{align}

Now observe that $OPT_2<S_2+ E[v_b|\hbox{instance is type 2}]\cdot \Pr[\hbox{instance is type 2}]$
since the RHS is the expected welfare from type $2$ instances when the buyer always gets the item and the seller keeps the item whenever $v_s>v_b$ (so we sometimes count the value of \emph{both} bidders). Thus the RHS is clearly at least than $OPT_2$, which is simply $\max(v_b,v_s)$ in type $2$ instances.
Together with (\ref{eq:bilateral-2-approx-1}) we get that $ALG_2 \geq \frac{OPT_2}{2}$.
\end{proof}
This completes the analysis of the approximation ratio of the median mechanism.\qed
\end{proof}

One nice feature of the median mechanism is that it only requires the knowledge of the distribution of one of the agents (the seller). We show that the median mechanism achieves the best approximation ratio among such mechanisms (proof in the appendix):

\begin{proposition}\label{prop-bilateral-distribution-info}
No truthful, individually rational and budget balanced mechanism an obtain an approximation ratio better than $2$ if the mechanism uses only the distribution of the seller or only the distribution of the buyer.
\end{proposition}

As we show in the next theorem (proved in the appendix), carefully choosing a trade price that depends on \emph{both} distributions enables us to get a better approximation ratio:
\begin{theorem}
There exists a truthful, individually rational, and budget balanced mechanism that achieves a $\frac{55}{28}$-approximation to the optimal social welfare.
\end{theorem}

Myerson and Satterthwaite \cite{MS83} already proved that no mechanism can be fully efficient. We further provide a quantification for that statement using the observation that every truthful deterministic\footnote{
Choosing a price from some distribution in this Bayesian environment cannot help, as there always exists a deterministic mechanism with at least the same welfare (i.e., the best price in the support).
} mechanism for bilateral trade simply sets a fixed trade price $r$, since the trade price cannot depend neither on the seller's value nor on the buyer's. Trade occurs if and only if $v_s\leq r$ and $v_b\geq r$. The simple proof is in the appendix.

\begin{proposition}\label{prop-bilateral-lowerbound}
No truthful, individually rational, budget balanced mechanism can achieve an approximation ratio better than $1.1231$ to the social welfare.
\end{proposition}

\subsection{Dissolving Partnerships}

In the \emph{partnership dissolving} problem, there are $n$ agents, each agent $i$ owns a share $r_i$ of an asset, and $\sum_{i=1}^n r_i =1$. Each agent $i$ has a value $v_i$ for holding the entire asset, or a value $c\cdot v_i$ for holding a fraction $c\geq 0$ of the asset. Let $r_{max}=\max\{r_1,...,r_n\}$ be the largest share held by an agent.

We first construct a mechanism that does not use any distributional assumptions, yet it provides a good approximation ratio given that the initial endowment of every player is not ``too big''. As we will see, achieving a better approximation ratio requires some distributional assumptions.

We then develop mechanisms with improved guarantees if the valuations are drawn from some known distirbutions. We show that any approximation algorithm for bilateral trade can be used for constructing a mechanism for partnership dissolving with the same approximation ratio.


\subsubsection{The Pivot Mechanism: Prior-Free Partnership Dissolving}

We first present a simple prior-free approximation mechanism for the partnership dissolving problem. The mechanism announces a price equal to the second highest value, and all bidders may sell their share at that price to the highest-value bidder. As long as the agent with the second-highest value does not own a large share of the asset, the mechanism achieves a good approximation.

Without loss of generality, assume that $v_1 \geq v_2 \geq ... \geq v_n$. The \emph{Pivot Mechanism for partnership dissolving} is defined as follows:
\begin{itemize}
\item Each bidder $i\geq 3$ give their shares $r_3,...,r_n$ to bidder $1$. Bidder 2 keeps his share.
\item Bidder $1$ pays $r_i\cdot v_2$ to every bidder $i\in \{3,...,n\}$.
\end{itemize}

\begin{proposition}
The Pivot Mechanism for partnership dissolving is truthful, individually rational, and budget balanced. It recovers a fraction of $1-r_{max}$ of the optimal social welfare in every instance.
\end{proposition}
\begin{proof}
It is straightforward to see that the mechanism is truthful, individually rational and budget balanced. We now analyze the approximation ratio of the mechanism. Denote the optimal welfare by $OPT$ and observe that $OPT=v_1$. The welfare of the mechanism is:
\begin{align*}
(1-r_2)\cdot v_1+r_2\cdot v_2
\geq   (1-r_2)\cdot v_1
\geq  (1-r_{max})OPT
\end{align*}
\end{proof}

We note that, as stated, the mechanism only partially dissolves partnerships, since sometimes both players $1$ and $2$ end up with non-zero shares $x_1$ and $x_2$. Mechanisms with partial dissolving can be turned to full dissolving as follows: denote the payment of agent $i$ by $p_i$. Run the mechanism, and give all shares to player $i$ ($i=1,2$) with probability $t_i=\frac {x_i} {x_1+x_2}$ for a payment of $p_i/t_i$. The approximation ratio is the same, the mechanism is still individually rational, but now the mechanism is only \emph{truthful in expectation} (notice that the expected payment of each bidder $i$ is still $p_i$).

We now show that without distributional knowledge, the pivot mechanism for partnership dissolving is essentially the best we can get.

\begin{proposition}\label{prop-partnership-lower}
\label{prop:partnership-hardness-rmax}
No deterministic, truthful, individually rational, and budget balanced mechanism recovers more than $1-r_{max}$ of the optimal social welfare in every instance, for every $r_{max}\geq \frac 1 2$.
\end{proposition}
\begin{proof}
Consider an instance with $n$ bidders. Let the valuations of bidders $3,\ldots, n$ be identically zero and $r_3=\ldots=r_n=0$. Let $r_1=1-r_{max}$ and $r_2=r_{max}$. Observe that since the valuations of bidder $3,\ldots, n$ are $0$ and they have no shares, the only bidders $1$ and $2$ may trade. By dominant-strategy truthfulness, the trade price cannot depend on $v_1$ and $v_2$, hence it is some fixed price $p$ for every $v_1$ and $v_2$. Now assume that $v_1=p-\epsilon$ and $v_2=0$. The optimal solution has value $v_1$ (give bidder $1$ all shares). However, no trade occurs so the value of the outcome is only $(1-r_{max})\cdot v_1$.
\end{proof}

\subsubsection{From Bilateral Trade to Partnership Dissolving }
\label{subsec:partership-reduction}


Next we show a reduction from partnership dissolving to bilateral trade. Given distributional assumptions on the values of the bidders the reduction guarantees constant approximation ratios for partnership dissolving irrespectively of the size of the initial shares. The idea is that each agent sells his share to the other agents via the bilateral trade mechanism. However, the mechanism needs to carefully adjust the prices so that truthfulness is maintained.

\begin{lemma}
\label{lem:reduction-to-n-buyers}
Let $M$ be some truthful, individually rational, and budget balanced mechanism for bilateral trade that achieves an $\alpha$-approximation to the welfare. There is a truthful, individually rational and budget balanced mechanism for partnership dissolving which also achieves an $\alpha$-approximation to the welfare.
\end{lemma}
\begin{proof}
In the proof we use the already-mentioned fact that any truthful mechanism for bilateral trade simply sets a trade price $p$ that does not depend on the values of the bidders. We develop our mechanism for partnership dissolving in two stages.

\vspace{0.1in} \noindent \textbf{First Stage: A Mechanism where only Bidder $i$ may sell.} We will ``run'' $M$ with bidder $i$ and a hypothetical buyer whose value is distributed according to the distribution of $\max_{k\neq i} v_k$. Let $p$ be the price that $M$ posts. Let $M'$ be the following mechanism:
\begin{itemize}
\item Let $j \in \arg\max_{k\neq i} \{v_k\}$ and let $m_2$ be the second highest value of $v_1,\ldots, v_{i-1},v_{i+1},\ldots, v_n$, that is,
$
    m_2=\max_{k \in N\setminus \{j,i\}} v_{k}.
$
\item Let $p^*=\max\{ p,m_2 \}$. If $v_i\leq p^*$ and $v_j \geq p^*$ then bidder $j$ pays to bidder $i$ the amount of $p^*$ and receives the item. Otherwise bidder $i$ keeps his item and does not get paid.
\end{itemize}

To see that the mechanism is truthful, observe that if a sale is made neither the winning buyer nor bidder $i$ cannot affect the price by changing their bid. A losing buyer $k$ can only turn into a winner by overbidding the winning buyer and paying $p^* \geq v_j\geq v_k$, and therefore cannot gain a positive payoff. If the item is not sold, it is either because the seller's value exceeds $p^*$ (and winning by underbidding induces a payment below $v_i$) or all of the buyers' values are below $p^*$ (and again, overbidding results in a payment higher than $v_j$).

Now for the approximation ratio. Notice that whenever there is a trade in $M$ there is a trade in $M'$ (but the opposite is not true; for example, a trade takes place when $p<v_i\leq m_2$).
$M$ achieves an $\alpha$-approximation to the optimal solution (that may only allocate bidder $i$'s share) which equals $\max\{v_i,\max\{v_{-i}\}\}=\max_{k}\{v_k\}$. $M'$ achieves at least the same expected welfare as $M$, thus it is an $\alpha$-approximation to the optimal welfare (that may only allocate bidder $i$'s share) as well.

The mechanism is individually rational since we always have that $v_i \leq  p^* \leq v_j$. In addition, payments are transferred from one player to another, hence the mechanism is budget balanced.

\vspace{0.1in} \noindent \textbf{Second Stage: The Final Mechanism.} At an arbitrary order, use $M'$ to sell to the other bidders the endowment $r_i$ of each bidder $i$ as a single indivisible item.

We now analyze the approximation ratio. Let $v_{max}=\max_kv_k$ be the highest value. By selling the endowment of bidder $i$ the expected social welfare is at least $\frac{r_i v_{max}}{\alpha}$. Since the valuations of the bidders are linear, after selling all endowments the expected social welfare of at least $\frac{v_{max}}{\alpha}$.

The truthfulness of the mechanism also follows from the linearity of the valuations of the bidders: at every stage they will maximize their payoff from the item independently of the other sales. Therefore, truthfulness follows from the truthfulness of $M'$. Similarly, the mechanism is individually rational and budget balanced.
\end{proof}

As our best approximation for bilateral-trade is $\frac{55}{28}$, the reduction guarantees the same approximation ratio for partnership-dissolving.

\begin{corollary}
There is a truthful, individually rational, and budget balanced mechanism the Partnership Dissolving problem which is a $\frac{55}{28}\approx 1.964$ approximation to the optimal social welfare.
\end{corollary}

\section{Combinatorial Exchanges}\label{sec-combinatorial}

In this section we consider \emph{combinatorial exchanges}: there are $n$ agents, each agent $i$ initially holds a subset $E_i$ of the items. Items are heterogeneous and indivisible. Every agent $i$ has a subadditive valuation $v_i$, 
that is, for every two bundles $S,T$ we have that $v_i(S\cup T)\leq v_i(S)+v_i(T)$.
Each $v_i$ is independently drawn from a distribution $F_i$. However, our mechanism will only require that the mechanism knows, for each agent $i$, the median value $MED_{i}$ for the bundle $E_i$ she initially owns.

Let $H_n$ be the $n$'th harmonic number ($H_n=\sum_{i=1}^n\frac{1}{i}$) and $t=\max_i|E_i|$. We present a mechanism that achieves an $8H_t$ approximation in this multi-parameter domain. In particular, if each player is initially endowed with at most one item we get an $8$-approximation. 

The basic step of the mechanism is in some sense a reduction to the bilateral trade problem: the bidders are randomly partitioned into two sets, ``buyers'' and ''sellers'', and each ``seller'' $i$ is offered to sell his endowment bundle at a price $MED_i$. Then we would like to take all the items that were sold and optimally allocate them among the ``buyers'' using VCG. The main obstacle is that VCG is not budget balanced. To overcome this we present a procedure that guarantees (approximate) welfare maximization while guaranteeing a minimum amount of revenue. Subsection \ref{subsec-global-reserve} describes this procedure and the mechanism itself is in Subsection \ref{subsec-combinatorial-median}.

\subsection{Detour: Combinatorial Auctions with Global Reserve}\label{subsec-global-reserve}

Consider the usual combinatorial auction setting, where a set $M$ of $m$ heterogeneous items that has to be allocated to $n$ bidders, each bidder $i$ has a valuation $v_i:2^M\rightarrow \mathbb R$. As usual we assume that each $v_i$ is normalized ($v_i(\emptyset)=0$) and non-decreasing. While the standard goal in the literature is to maximize welfare, assuming the auctioneer has no production cost for the items, in our case the auctioneer is only interested in selling the items to cover his non-negative cost $r$ of producing all items\footnote{The items are produced only if a sale is made. Here we assume for simplicity that the cost of producing the first item is $r$ and the cost of producing any additional item is $0$. This corresponds to the case that the production cost of items is governed by the start-up cost. A more realistic setup assumes a production cost for each item, or more generally for bundles of items. Indeed, the mechanism of Subsection \ref{subsec-combinatorial-median} essentially presents a solution for this case.}. We are interested in truthful and individually rational mechanisms.

Denote by OPT the value of the welfare-maximizing allocation of the items in $M$ to the bidders. Clearly, if $OPT<r$ then no such individually rational mechanism is possible. If $OPT=r$ then again no such mechanism exists but now the argument is a bit more delicate: by revenue equivalence, VCG is the only truthful and individually rational mechanism that always outputs an optimal solution, and it is easy to construct examples when its revenue is $0$. For example, consider two bidders and a set of two items $a$ and $b$. The value of bidder $1$ for any bundle that contains $a$ is $1$ (the value of any other bundle is $0$). Similarly the value of bidder $2$ for bundles that contain $b$ is $1$. VCG will allocate $a$ to bidder $1$ and $b$ to bidder $2$, but none of the bidders will pay anything.

We therefore relax our requirements: given $\alpha>1$, whenever $OPT\geq \alpha \cdot r$ the algorithm must allocate some items to the bidders and raise a revenue of at least $r$. Else, when $OPT< \alpha \cdot r$ the mechanism is not required to sell the items (but if it does sell the revenue must be at least $r$).

The challenge is of course to develop such a mechanism with $\alpha$ that is as small as possible, and we do so for $\alpha=H_n$. We use the (well known) observation that VCG generalizes to maximization of an affine function. Specifically, we ``adjust'' the welfare of an allocation $A=(A_1,\ldots, A_n)$ to be $\Sigma_iv_i(A_i)-H_{n_A}\cdot r$, where $n_A$ is the number of non-empty bundles in $A$. We now select the allocation with the highest ``adjusted welfare''. Payments are similar to VCG payments.

The mechanism is truthful since VCG is truthful. In addition, the mechanism allocates the items whenever $OPT\geq \alpha \cdot r$. To see this, observe that the mechanism allocates some items only if there is an allocation with a positive adjusted welfare (since the adjusted welfare of the empty allocation is $0$). Now recall that the adjusted welfare of the optimal allocation is at least $OPT-H_n\cdot r$.

All that is left is to prove that when the mechanism allocates some items then the revenue is at least $r$. Suppose that the mechanism outputs the allocation $A=(A_1,\ldots, A_n)$. Consider some bidder $i$ with $A_i\neq \emptyset$. In the VCG mechanism bidder $i$ pays his ``damage to society''. Observe that the damage to society of bidder $i$ is at least $\frac r {H_{n_A}}$: consider the allocation $A'$ that allocates each bidder $i'\neq i$ the same set of items and allocates nothing to bidder $i$. If we ignored the preferences of bidder $i$, we could have chosen the allocation $A'$ and increase the adjusted welfare by
$$
(\Sigma_{i'\neq i}v_{i'}(A_{i'})-H_{n_A-1}\cdot r)-(\Sigma_{i'\neq i}v_{i'}(A_{i'})-H_{n_A}\cdot r)=H_{n_A-1}\cdot r-H_n\cdot r=\frac r {n_A}
$$
I.e., the payment of each bidder $i$ with $A_i\neq \emptyset$ is at least $\frac r {n_A}$. Since by definition there are at least $n_A$ such bidders, the total revenue is at least $r$, as required. 

\subsection{The Combinatorial Median Mechanism}\label{subsec-combinatorial-median}

Borrowing ideas from the procedure for combinatorial auction with a global reserve and the median mechanism, we get the \emph{Combinatorial Median} mechanism for combinatorial exchanges.

\begin{enumerate}
\item Each player is assigned to either group $S$ or group $B$ uniformly at random.
\item Each player $i\in S$ will be offered a price equal $MED_{i}$ (i.e., the median value of her endowment). Let $\hat S$ denote the set of players in $S$ that accepted the price. The total set of items of players in $\hat S$ endowments is denoted by $E_{\hat{S}}=\cup_{i\in \hat S}E_i$.
\item Given an allocation $A$ of items in $E_{\hat{S}}$ to players in $B$, denote for each $i\in \hat S$ by $t_i$ the number of buyers that hold in $A$ at least one item from $E_i$, i.e., $t_i=|\{j|A_j\cap E_i\neq \emptyset \}|$. Let $c_A=\Sigma_{i\in \hat S} H_{t_i} \cdot MED_i$.
\item Run a VCG auction for the items $E_{\hat{S}}$ among the bidders $B$ where we penalize the welfare of an allocation $A$ by $c_A$, taking into account the endowments of bidders in $B$. I.e., we find the allocation $A$ that maximizes: $\Sigma_{i\in B}v_i(A_i+E_i)-c_A$.
\item Consider seller $i\in \hat{S}$. If at least one item from his endowment $E_i$ is sold in the VCG auction, then $i$ is paid $MED_{i}$ and loses all his endowment. Else, seller $i$ keeps his endowment and is not paid anything. Each buyer is allocated the items he won in the VCG auction (in addition to his endowment) and pays his VCG payment.
\end{enumerate}

The role of the $c_A$'s is to guarantee budget balance in a slightly more complicated way than was described earlier. This introduces inefficiency to the market, but we show that this loss is bounded.

\begin{theorem}
The Combinatorial Median mechanism is a $8H_t$-approximation to the optimal social welfare. It is truthful, ex-post individually rational and weakly budget balanced.
\end{theorem}
\begin{proof}
In the analysis we use the following simple folklore observation:
\begin{claim}
\label{obs:subadditive-random}
Let $v$ be a subadditive valuation. Let $S$ be a set and let $R$ be a randomly constructed set that is obtained by selecting each item in $S$ with probability $\frac 1 2$, uniformly at random. Then, $E[v(R)]\geq \frac{1}{2}v(S)$.
\end{claim}
\begin{proof}
By subadditivity and the random choice of $R$:
$
E[v(R)]  \geq E[v(S)-v(S \setminus R)] \geq v(S)-E[v(R)]
$.
\end{proof}

The proof uses the following notation: given a set of players $T$ and a set of items $R$, let $OPT_{T \leftarrow R}$ be the optimal allocation of items in $R$ to players in $T$. Given an allocation $A$, let $A|_{T,R}=\Sigma_{i\in T}v_i(R\cap A_i)$ (the value of players in $T$ from items in $R$). 

\begin{claim}
The mechanism achieves in expectation at least $\frac 1 {2H_t}$ of the share of the optimal efficiency gained by bidders in $B$ receiving the items 
$E_{\hat{S}}$, that is:
\begin{align*}
ALG  \geq \frac{1}{2H_t} E_v  \big[ E_{S,B} \big[
OPT(v)|_{B,E_{\hat{S}}}
\big] \big]
\end{align*}
\end{claim}
\begin{proof}
Let $ALG_{S,B}$ be the expected efficiency of our mechanism after partitioning to $S,B$. We first argue that the effect of the $c_A$'s on the welfare maximizing allocation is limited: given $v,S,B$, our mechanism outputs the welfare-maximizing allocation of items in $E_{\hat{S}}$ to the bidders in $B$, taking the $c_A$'s into account. The efficiency of the outcome ais composed of three terms: the value of the allocation $A$, minus $c_A$, plus the value of the players $S\setminus \hat{S}$ that did not accept the median price and consume their own endowment.

The value of $c_A$ is at most $H_t$ times the sum of the medians of sellers that agreed to sell their items, and the value of players in $S\setminus\hat{S}$ is at least the sum of the medians offered to them. Since sellers belong to $\hat S$ independently with probability $\frac 1 2$, we conclude that, in expectation, the total value of the players in $S\setminus \hat{S}$ is at least $\frac 1 {H_t}$ of $c_A$. Let $A=OPT(v)|_{B,E_{\hat{S}}}$. We have that:
\begin{align}
ALG_{S,B} &\geq E_v \big[\max\{0,A -c_{A}\}]+\frac {E_v[c_{A}]} {H_t} \geq \frac {E_v [ A]} {H_t}\nonumber
\end{align}
Since $E_v \big[\max\{0,A -c_{A}\}]\geq E_v \big[A] -E_v[c_{A}]$
substituting  $E_v[c_{A}]$ instead of $E_{v}\big[A]$ yields the second inequality when $E_{v}\big[A]- E_v[c_{A}]\geq 0$.
When $E_{v}\big[A]- E_v[c_{A}]< 0$, we observe that $E_v \big[\max\{0,A -c_{A}\}]\geq 0$ and the second inequality follows.

Thus, all we have to do is to bound the value of the nominator (where $\sigma$ denotes an optimal allocation of items in $S$ to $B$, given a partition $S,B$)\footnote{In the analysis, we ignore the endowments that players in $B$ continue to hold, as it can only improve the performance of our mechanism.}:

\begin{align}
E_v \big[OPT(v)|_{B,E_{\hat{S}}}\big] & \geq   E_{v} \left[ OPT_{B \leftarrow E_{\hat{S}} }(v) \right] \geq  E_{v} \left[ OPT_{B \leftarrow E_{\hat{S}}}(v)|_{ B,E_{\hat{S}}} \right] \nonumber \\
& =  E_{v \textrm{ of B}}  \left[ \;\;
\sum_{i\in B }  E_{v \textrm{ of S}} \left[ \; v_i\left( \; \sigma_i \cap E_{\hat{S}} \; \right)\; \right]
\;\; \right]  \nonumber \\
& \geq  E_{v \textrm{ of B}}  \left[\sum_{i\in B } \frac{1}{2} v_i(\sigma_i)
\right]  = \frac{1}{2} E_{v} \left[ OPT_{B \leftarrow E_{\hat{S}}}(v)
\right]
 \geq \frac{1}{2}  E_{v} \left[
OPT(v)|_{B,E_{\hat{S}}}
\right] \nonumber
\end{align}
Where in the third-to-last transition we use Observation \ref{obs:subadditive-random}. As the preferences $v$ and the partition $S,B$ are drawn independently, it follows that:

\begin{align}
ALG & =  E_{S,B} \big[ ALG_{S,B} \big] \big]
& \geq \frac{1}{2H_t} E_{S,B}  \big[ E_{v} \big[
OPT(v)|_{B,E_{\hat{S}}}
\big] \big]
&= \frac{1}{2H_t} E_{v}  \big[ E_{S,B} \big[
OPT(v)|_{B,E_{\hat{S}}}
\big] \big]\nonumber
\end{align}
This concludes the proof of the lemma.
\end{proof}

We now show that the random partition to $S$ and $B$ generates an efficiency loss which is not greater than a factor of $4$. Given a profile $v$, let $\sigma^{v}$ be the optimal allocation for $v$.
Then,
\begin{align}
E_{S,B} \big[  OPT(v)|_{B,E_{\hat{S}}} \big] =& \sum_{i=1}^n \Pr\big( i \in B \big) \cdot E_{S,B} [  v_i(\sigma^v_i \cap E_{\hat{S}}) ]
\geq&  \sum_{i = 1}^n \frac{1}{2} \cdot \frac{1}{2} v_i(\sigma^v_i)
=& \label{eq:parition-S-B-1/4-approx}\frac{1}{4} OPT(v)\nonumber
\end{align}

Where the second transition is due to Observation \ref{obs:subadditive-random}. Together with the lemma we have that ALG is a $8H_t$-approximation to OPT.

The mechanism is clearly truthful, as the agents in $B$ participate in a VCG auction (more precisely, an affine maximizer), the agents in $S$ face take-it-or-leave-it offers, and agents are randomly assigned to $S$ and $B$.

Finally, we show that the mechanism is budget balanced. Let $A$ be the allocation that the VCG mechanism outputs (i.e., each player $i\in B$ receives $A_i+E_i$). Consider bidder $i\in B$ that receives some items that were initially endowed to some set of players $S'\in S$. Let $A'$ be the allocation in which $A_{i'}=A'_{i'}$ for all $i\neq i'$ and $A'_i=\emptyset$. Recall that in the VCG payment formula bidder $i$ pays his ``damage to society'' which is at least (comparing the welfare of $A$ and and the welfare of $A'$):
\begin{align*}
\Sigma_{i\neq i', i\in B} \left( v_i(A_i+E_i)-c_A \right) &-\Sigma_{i\neq i', i\in B} \left( v_i(A'_i+E_i)-c_{A'} \right) =c_A - c_{A'}\\
=& \Sigma_{i'\in S'} H_{t_{i'}} \cdot MED_{i'} - \Sigma_{i'\in S'} \cdot H_{t_{i'}-1} \cdot MED_{i'}\\
=& -\Sigma_{i'\in S'} \frac {MED_{i'}} {t_{i'}}
\end{align*}
In other words, we can think of each player $i$ that got at least one item from $E_{i'}$ as paying $\frac {MED_{i'}} {t_{i'}}$ to player $i'$ (player $i$ might pay an additional amount of money, but this additional amount is ``burned''). Since by definition there are $t_{i'}$ players that got at least one item from $E_{i'}$, the total payment that $i'$ gets is exactly $MED_{i'}$ which equals to the amount that $i'$ receives, as needed.
\end{proof}

\section{Arrow-Debreu Markets}\label{sec-multiunit}

In this section we give a constant approximation mechanism for a multi-parameter environment without any distributional assumptions. We consider the following setting \cite{AD54}: there is one divisible good and $n$ players. Each player $i$ has a valuation function $v_i:[0,1]\rightarrow \mathbb R$, and for every $x,y$ we define the marginal valuation $v_i(x|y)=v_i(x+y)-v_i(y)$.
 We assume that the valuation functions are normalized ($v_i(0)=0$),
non decreasing, 
and have decreasing marginal valuations (i.e., $v_i(\epsilon|x)\geq v_i(\epsilon|y)$ for every $\epsilon>0,y>x$).\footnote{
When $v_i(\cdot)$ is twice differentiable, we simply assume that $v_i^{''}(x) \leq 0$ for every $x$.
} 
Denote the initial endowment of player $i$ by $r_i$, $r_i\geq 0$, where $\Sigma_i r_i=1$. Observe that given some price $p$, the supply that a seller $i$ is willing to sell is an amount $x_i$ that maximizes his payoff $p \cdot x_i -v_i(x_i|r_i-x_i)$.

Our mechanism is in many respects a varaiant of the Combinatorial Median mechanism for combinatorial exchanges. Similarly, players are divided into ``sellers'' and ``buyers'' (but in a subtler way). The constant approximation ratio is achieved by replacing the revenue extraction procedure of the combinatorial exchanges mechanism with a method that allows the separate sell of items. This will be the key to obtaining a constant approximation in the worst case. We assume that no single player initially holds a huge chunk of the good; Notice that if, say, one player is initially endowed with all the good then no prior-free mechanism can achieve a bounded approximation.

\begin{theorem}
Suppose that for all $i$, $r_i\leq \frac 1 3 $. Then, there exists a
truthful, weakly budget balanced mechanism that provides an expected approximation ratio of $48$ in every instance.
\end{theorem}

We first provide a mechanism assuming that the bidders can be divided to $3$ groups $N_1$, $N_2$, $N_3$ where each set $N_k$ is \emph{substantial}: $\Sigma_{i\in N_k}r_i\geq \frac 1 4$. Later we relax the requirement; we will only assume that for every $i$,  $r_i\leq \frac 1 3 $. We need the following definition:
\begin{definition}
Let $N_k$ be a substantial set of bidders. The \emph{mid-supply price} of $N_k$ is the minimal price $p$ such that the total amount that bidders in $N_k$ are willing to sell at price $p$ is at least $\frac 1 8$.
\end{definition}

The mechanism itself is a bit heavy on details, although the basic idea is quite simple. We therefore start with an informal description, and then move on to a formal one. Initially, we have an arbitrary division of the bidders into three arbitrary substantial groups, $N_1$, $N_2$, and $N_3$, with roles selected at random: $N_1$ wil be the group of buyers, $N_2$ the statistics, and $N_3$ the sellers.

We use the statistics group to compute a mid-supply price $p$: that is, the price for which bidders in the statistics group are willing to sell half of their total supply (which is at least $\frac 1 8$ of the good). Each bidder in $N_3$ is asked to report the maximum amount of his endowment that he is willing to sell at price $p$. Let $t$ be the total amount that the sellers are willing to sell (in fact, we have to make sure that $t\leq \frac 1 8$ -- see the formal description for exact implementation details). 

Now, run VCG with the participation of the buyers and an extra additive buyer with valuation $v_d(s)=\min(t,s )\cdot p$. 
This bidder is added to ensure compliance with the budget balance requirement.
The set of possible allocations in this VCG mechanism equals to all distributions of amount $t$ of the good among the buyers and the extra additive buyer. Effectively, we show that this amounts to finding a welfare maximizing allocation of $t$ fraction of the good among the buyers so that each buyer that received an amount of $x$ pays at least $x\cdot p$. We use this money to pay each bidder $i\in N_3$ a total sum of $x_i\cdot p$, where $x_i$ is the part of the endowment that was taken from bidder $i$. We now provide a formal description of the mechanism, followed by its analysis.

\subsection*{The Formal Mechanism:}
\begin{enumerate}
\item Let $N_1$, $N_2$ and $N_3$ be three substantial groups of bidders. 

\item   Select at random ``roles'' for the groups $N_1$, $N_2$, $N_3$: players in one group will be the \emph{buyers} (without loss of generality, $N_1$), another group will be the \emph{statistics} group (without loss of generality, $N_2$), and players in the additional group are the \emph{sellers} ($N_3$). 

\item \label{item:multi-unit-mechanism-pricing-stage} Let $p$ be the mid-supply price of the statistics group $N_2$. Each seller from the sellers group $N_3$ reports the amount of good $x'_i$ he is willing to sell at price $p$. Let $t=\min\{\frac 1 8,\Sigma_{i\in N_3}x_i\}$.

If $\Sigma_{i\in N_3}x'_i\leq \frac 1 8$, let $x_i=x'_i$. If $\Sigma_{i\in N_3}x'_i>\frac 1 8$, choose a value $x_i$ for each $i$ such that $x_i\leq x'_i$ and $\Sigma_{i\in N_3}x_i=\frac 1 8$.\footnote{ Formally, order the buyers arbitrarily, and let 
$x_i=\max\{0,\frac 1 8 - \Sigma_{i'>i}x'_{i'}\}$ if $x'_i+\Sigma_{i'>i}x'_{i'}\geq \frac 1 8$.}

\item For each player $i\in N_1$, let $v'_i(s)=v_i(s|r_i)$.
Let $N'$ be a set that consists of all players in $N_1$ and one additional dummy bidder. 

Use the VCG mechanism to
sell $t$ fraction of the good to buyers in $N'$, where the valuation of each player $i\in N_1$ is $v'_i$ and the valuation of the dummy bidder is $v_d(s)=\min(t,s )\cdot p$.

\item \label{item:multi-unit-mechanism-dummy-stage} The output of the mechanism is as follows: each bidder $i\in N_1$ pays to the mechanism the VCG payment of $v'_i$ and receives the same amount of good that $v'_i$ received (in addition to his endowment $r_i$). Bidders in $N_2$ keep their initial endowment and do not pay anything.

Let $t'\leq t$ be the amount of good that the dummy player ended up with in the VCG mechanism. Choose some $x''_i$'s such that for each $i\in N_3$, $x''_i\leq x_i$ and $\Sigma_{i\in N_3}x''_i=t-t'$, taking the good first from bidders with lower indices. Each bidder $i\in N_3$ keeps $r_i-x''_i$ of the good and receives a payment of $x''_i\cdot p$.
\end{enumerate}

\begin{claim}
The above mechanism is truthful.
\end{claim}
\begin{proof}
The mechanism is clearly truthful for the statistics group since they never sell nor receive any amount of the good. The mechanism is also truthful for the buyers since they are just participating in a VCG mechanism. To show that the mechanism is truthful for the sellers we have to use the fact that the valuations exhibit decreasing marginal utilities. First, observe that the price $p$ depends only on the valuations of the statistics group $N_2$.
Now, consider some seller $i \in N_3$. Bidder $i$ reports at stage \ref{item:multi-unit-mechanism-pricing-stage} a quantity $x'_i$ that maximizes his profit $x\cdot p+v_i(r_i-x)$. Therefore, if he eventually sells a fraction $x'_i$ he has no reason to report a different value. However, at stages \ref{item:multi-unit-mechanism-pricing-stage} and \ref{item:multi-unit-mechanism-dummy-stage} of the mechanism the quantity that he sells is reduced to $x_i$ or $x''_i$ that may gain him a lower profit. By the way that the mechanism reduces the quantities, reporting any value above $x''_i$ will not affect the quantity that seller $i$ sells. If seller $i$ reports a value smaller than $x''_i$, he will sell a quantity smaller than $x''_i$; This smaller quantity cannot gain him a greater profit since the profit is non-decreasing in $x$ in the range below $x'_i$ due to decreasing marginals.\footnote{
To see this, note that the derivative of
$x\cdot p$ is $p$ for every $x$.
Since $x'_i$ maximizes profit, and due to the convexity of $v_i$,
for every value $x<x'_i$ and the derivative of $v_i(r_i-x)$ is negative with absolute value of at most $p$. Therefore, the marginal profit is non-negative for $x<x'_i$. A similar argument holds also when $v_i$ is not differentiable.}
\end{proof}

\begin{claim}
The above mechanism is weakly budget balanced.
\end{claim}
\begin{proof}
Consider a buyer $i$ that received an amount of $t_i$ (not including his endowment $r_i$). His VCG price is at least $t_i\cdot p$, since otherwise we could have considered the same allocation except that an additional amount of $t_i$ of the good is allocated to the dummy bidder. We have that the total payment is at least $(t-t')\cdot p$, which is exactly the amount we have to pay to the sellers.
\end{proof}

The next lemma analyses the approximation ratio of the mechanism:

\begin{lemma}
If there are three substantial groups $N_1$, $N_2$, and $N_3$ then the mechanism provides an approximation ratio of $48$.
\end{lemma}
\begin{proof}
Fix some optimal solution $(o_1,\ldots, o_n)$. For $k=\{1,2,3\}$, let $O_k=\Sigma_{i\in N_k} v_i(o_i)$. Observe that since each group of bidders plays the role of the buyers with probability exactly $\frac 1 3$, we have that $E[O_1]= OPT/3$. Let $p'$ be the mid-supply price of $N_3$ and recall that $p$ is the mid-supply price of $N_2$. Since the statistics group and the sellers group are chosen at random, with probability at least $1/2$ we have that $p\geq p'$. We will condition our analysis on that event and conservatively assume that if $p<p'$ then the welfare of the allocation that the mechanism outputs is $0$.

Now, if $p\geq p'$, the total amount of the good that bidders in $N_3$ are willing to sell is at least $\frac 1 8$: they hold at least $\frac 1 4$ of the good by our initial condition, and at price $p'$ they are willing to sell half of it, so surely they will agree to sell that amount at price $p\geq p'$. In particular we have that $\Sigma_ix_i= \frac 1 8$ with probability at least $\frac 1 2$.

\begin{claim}
For every $i\in N_1$, let $s_i$ denote the amount of good bidder $i$ receives in the final allocation. If $t=\frac 1 8$ then $\Sigma_{i\in N_1}v_i(s_i) \geq \frac {O_1} 8-\frac p  8$.
\end{claim}
\begin{proof}
Consider the allocation that gives each bidder $i\in N_1$ an amount of $s'_i=\frac {o_i} 8$. Since the valuations have decreasing marginals, the welfare of this allocation is at least $\Sigma_iv_{i\in N_1}(s'_i)\geq \frac {O_1} 8$ and no more than $\frac 1 8$ of the good was allocated to players in $N_1$.

Thus, the welfare of VCG is at least the welfare of $(s'_1,\ldots ,s'_n)$, but we have to subtract the contributtion of the dummy bidder. His contribution to the welfare is at most his maximum value: $\frac p 8$. Thus we have that
$\Sigma_{i\in N_1}v_i(s_i) \geq \frac {O_1} 8-\frac p  8$.
\end{proof}

Now notice that the value of bidders in the statistics group $N_2$ is at least $p\cdot t=\frac p 8$: they are not willing to sell $\frac 1 8=t$ of the good at price $p$, so their total value for their initial endowment (that they keep) is at least $p\cdot t$.

Hence we have that with probability at least $\frac 1 2$ (if $t=\frac 1 8$), it holds that $\Sigma_{i\in N_1}v_i(s_i)+\Sigma_{i\in N_2}v_i(s_i) =\frac {O_1} 8-\frac p  8 + \frac p 8=\frac {O_1} 8$. Recall that $E[O_1]=\frac 1 3$, and we get that the expected welfare is at least
$\frac 1 2 \cdot \frac {E[O_1]} 8 \geq \frac {OPT} {48}$,  
as needed.\qed
\end{proof}

We now relax the requirement of substantial groups and show how the above mechanism can be used for the more general case. The proof of the lemma is in the appendix.

\begin{lemma}\label{lemma-multiunit-no-substantial}
If for every $i$ we have that $r_i\leq \frac 1 3$ then there exists a mechanism that provides an approximation of $48$.
\end{lemma}

\bibliography{partnership}

\begin{thebibliography}{10}

\bibitem{Ake70}
George~A Akerlof.
\newblock The market for 'lemons': Quality uncertainty and the market
  mechanism.
\newblock {\em The Quarterly Journal of Economics}, 84(3):488--500, 1970.

\bibitem{AD54}
Kenneth~J. Arrow and Gerard Debreu.
\newblock Existence of an equilibrium for a competitive economy.
\newblock {\em Econometrica}, 22(3):pp. 265--290, 1954.

\bibitem{ALMS13}
Lawrence Ausubel, Jon Levin, Paul Milgrom, and Ilya Segal.
\newblock Incentive auction rules option and discussion, 2013.
\newblock Appendix to the FCC’s 28-Sep-2012 NPRM on Incentive Auctions.

\bibitem{BBK08}
Moshe Babaioff, Patrick Briest, and Piotr Krysta.
\newblock On the approximability of combinatorial exchange problems.
\newblock In {\em SAGT}, pages 83--94, 2008.

\bibitem{CGK87}
Peter Cramton, Robert Gibbons, and Paul Klemperer.
\newblock Dissolving a partnership efficiently.
\newblock {\em Econometrica}, 55(3):615--32, 1987.

\bibitem{CS06}
Martin~W. Cripps and Jeroen~M. Swinkels.
\newblock Efficiency of large double auctions.
\newblock {\em Econometrica}, 74(1):47--92, 2006.

\bibitem{Dob11}
Shahar Dobzinski.
\newblock An impossibility result for truthful combinatorial auctions with
  submodular valuations.
\newblock In {\em Proceedings of the 43rd Annual ACM Symposium on Theory of
  Computing}, STOC '11, pages 139--148, 2011.

\bibitem{DV12}
Shahar Dobzinski and Jan Vondr{\'a}k.
\newblock The computational complexity of truthfulness in combinatorial
  auctions.
\newblock In {\em ACM Conference on Electronic Commerce}, pages 405--422, 2012.

\bibitem{DV11}
Shaddin Dughmi and Jan Vondr{\'a}k.
\newblock Limitations of randomized mechanisms for combinatorial auctions.
\newblock In {\em FOCS}, pages 502--511, 2011.

\bibitem{DRT14}
Paul Dutting, Tim Roughgarden, and Inbal Talgam-Cohen.
\newblock Modularity and greed in double auctions.
\newblock {\em EC'14}.

\bibitem{FMS07}
Drew Fudenberg, Markus Mobius, and Adam Szeidl.
\newblock Existence of equilibrium in large double auctions.
\newblock {\em Journal of Economic Theory}, 133(1):550--567, 2007.

\bibitem{KSS04}
A.~Kothari, T.~Sandholm, and S.~Suri.
\newblock Solving combinatorial exchanges: optimality via a few partial bids.
\newblock In {\em Proceedings of the Third International Joint Conference on
  Autonomous Agents and Multiagent Systems, 2004.}, pages 1418--1419, 2004.

\bibitem{mcA92}
Preston~R. McAfee.
\newblock Amicable divorce: Dissolving a partnership with simple mechanisms.
\newblock {\em Journal of Economic Theory}, 56(2):266--293, 1992.

\bibitem{mcA08}
Preston~R. McAfee.
\newblock The gains from trade under fixed price mechanisms.
\newblock {\em Applied Economics Research Bulletin}, 1, 2008.

\bibitem{SM13}
Paul Milgrom and Ilya Segal.
\newblock Deferred-acceptance heuristic auctions, 2013.
\newblock Working paper, Stanford University.

\bibitem{Mol01}
Benny Moldovanu.
\newblock How to dissolve a partnership.
\newblock {\em Journal of Institutional and Theoretical Economics (JITE)},
  158(1), 2002.

\bibitem{M99}
Herv{\'e} Moulin.
\newblock Incremental cost sharing: Characterization by coalition
  strategy-proofness.
\newblock {\em Social Choice and Welfare}, 16(2):279--320, 1999.

\bibitem{MS83}
R.~B. Myerson and M.~A. Satterthwaite.
\newblock Efficient mechanisms for bilateral trading.
\newblock {\em Journal of Economic Theory}, 28:265 -- 281, 1983.

\bibitem{RS06}
Tim Roughgarden and Mukund Sundararajan.
\newblock New trade-offs in cost-sharing mechanisms.
\newblock In {\em STOC}, pages 79--88, 2006.

\bibitem{RSW94}
Aldo Rustichini, Mark~A Satterthwaite, and Steven~R Williams.
\newblock Convergence to efficiency in a simple market with incomplete
  information.
\newblock {\em Econometrica}, 62(5):1041--63, 1994.

\bibitem{SW02}
Mark~A. Satterthwaite and Steven~R Williams.
\newblock The optimality of a simple market mechanism.
\newblock {\em Econometrica}, 70(5):1841--1863, 2002.

\end{thebibliography}

\appendix

\section{Missing Proofs of Section \ref{sec-bilateral}}

\subsection*{Proof of Proposition \ref{prop-bilateral-distribution-info}}
We start with the first part. Consider the following distribution $D_s$ of the seller: with probability $\frac 1 2$, $v_s$ gets a value (uniformly at random) in $(0,\epsilon)$. With probability $\frac 1 2$, $v_s$ gets a value (uniformly at random) between $(1,1+\epsilon)$. Observe that the median of $D_s$ is $1$.

There are two possible cases, depending on the trade price $r$:
\begin{enumerate}
\item $r \leq 1$: let $v_b=\infty$ with probability $1$. The optimal solution always sells the item to the buyer, but the mechanism will sell the item with probability $\frac 1 2$. We get an approximation of $2$ since $v_b >> v_s$.

\item $r > 1$: let $v_b=1-\epsilon$ with probability $1$. The value of the optimal solution is at least $1-\epsilon$ (always sell the item to the buyer). However, the mechanism sells the item only when $v_s \in (0,\epsilon)$, which happens with probability $\frac 1 2$. The approximation is $2$ also in this case.
\end{enumerate}

We now prove the second part. Consider the following distribution $D_b$ of the buyer: with probability $0.99$, $v_b$ gets a value (uniformly at random) in $(1,1+\epsilon)$. With probability $0.01$, $v_b$ gets a value (uniformly at random) between $(100,100+\epsilon)$.

There are several possible cases, depending on the trade price $r$.
\begin{enumerate}
\item $r \leq 1+\epsilon$: let $v_s=1+\epsilon$ with probability $1$. The optimal solution sells the item to the buyer with probability $\frac 1 100$ and the expected welfare is about $2$. The mechanism that post the price $r$ however will never sell the item and will generate an expected welfare of $1+\epsilon$.

\item $1+\epsilon< r \leq 100+\epsilon$: let $v_s=0$ with probability $1$. The expected value of the optimal solution is about $2$ (always sell the item to the buyer). However, the mechanism sells the item only when $v_b \in (100,100+\epsilon)$, so the expected welfare is about $1$. The approximation is $2$ also in this case.

\item $r>100+\epsilon$: let $v_s=0$. The expected value of the optimal solution is $2$, but the mechanism achieves welfare of $0$.
\end{enumerate}

\subsection*{Proof of Proposition \ref{prop-bilateral-lowerbound}}

Consider a setting where the seller has an exponential distribution $F_s(v_s)=1-e^{-sv_s}$ and the buyer's distribution is $F_b(v_b)=1-e^{-bv_b}$.

The expected social welfare is simply the maximum of two exponential variables. The minimum of two exponential variables is distributed exponentially with parameter $s+b$, therefore its expectation is $\frac{1}{s+b}$. We get that:
$E[max\{v_s,v_b\}]=E[v_s+v_b]-E[min\{v_s,v_b\}]=1/s+1/b-1/(s+b)$.

For a given trade price $p$, the welfare is:
\begin{align*}
& e^{-bp}\left( \left(1-e^{-sp}\right)\left(p+\frac{1}{b}\right)+
e^{-sp}\left(p+\frac{1}{s}\right)\right)+\left(1-e^{-bp}\right)\frac{1}{s}\\
\end{align*}
Where we use that for the exponential distribution, $E[x|x>p]=p+E[x]$.
The welfare is maximized when the partial derivative by $p$ equals zero, i.e., when $(b^2-s^2)e^{-sp}=b^2-b^2sp$. For $s=1$ and $b=1/2$, we get that the optimal price is about $1.603$ that gives welfare of $2.0775$ compared to optimal welfare of $\frac{7}{3}$.


\section{A $\frac {55} {28}$-Approximation Mechanism for Partnership Dissolving}

In this section we present a mechanism for bilateral trade that provides an approximation ratio which is strictly better than $2$. Denote by $M_b$ the median of the distribution of the buyer and by $M_s$ the median of the distribution of the seller.

\begin{theorem}
There exists a $\frac {55} {28}$-approximation for the bilateral trade problem.
\end{theorem}
For clarity of presentation we normalize the value of the optimal solution OPT to be $1$. The mechanism itself simply chooses a trade price $p$ according to the following rule:
\begin{itemize}
\item If $M_b\geq M_s$:
\begin{itemize}
\item If expected value of the seller is greater than $1/13$, then $p \in [M_s,M_b]$. Else, $p=3/13$.
\end{itemize}
\item If $M_b\geq M_s$:
\begin{itemize}
\item Let $p$ be either $M_s$ or $M_b$, whichever gives a better expected welfare.
\end{itemize}
\end{itemize}

One part of the proof relies on McAfee's $2$ approximation for the gain for trade \cite{mcA08}. In Appendix \ref{sec-combinatorial-gft} we provide an alternative combinatorial proof for this result. We will also use the analysis of this  combinatorial proof in the proof of the theorem.

In the proof we let $x$ denote the random variable denoting the value of the item for seller, and $y$ be random variable denoting the value of the item for the buyer. We divide the proof to two cases, depending on the relationship between $M_b$ and $M_s$.

\subsection*{Case 1: $M_b\geq M_s$}\label{subsec-b-is-bigger}

We show that in this case we get an approximation ratio of $\frac {13} 7$. Let $r=E[x]$ (the expected value of the seller). We present two mechanisms that provides different approximation ratios that depend on $r$, and show that at least one them gives the guaranteed approximation ratio.

\begin{lemma}\label{lemma-mcafee-welfare}
If $M_b\geq M_s$ then the expected welfare provided by McAfee's mechanism is $\frac 1 2 +\frac r 2$.
\end{lemma}
\begin{proof}
McAfee \cite{mcA08} shows that $ALG-r \geq  \frac 1 2 (OPT-r)$. This implies that $ ALG \geq  \frac {OPT} 2 +\frac r 2$, as needed.
\end{proof}

The second mechanism is the following: for $t>1$, given $r=E[x]$ as defined above, the \emph{$t$-threshold} sets a trade price of $t\cdot r$. The next lemma analyzes the approximation ratio of the $t$-threshold mechanism (the analysis holds for all $M_s$ and $M_b$):

\begin{lemma}\label{lemma-small-r}
If $t < \frac {1-r} {2r}$ then the expected welfare of the allocation that the $t$-threshold mechanism produces is at least $1+r-\frac 1 t+\frac r t-t\cdot r$.
\end{lemma}
\begin{proof}
First observe that it trivially holds that $OPT \geq E[x]+E[y]$. Hence, $E[y]\geq 1-r$. Let $p=\Pr[x< t\cdot r]$, and let $q=\Pr[y< t\cdot r]$. By Markov inequality, $p\geq 1-1/t$. We can now write:
\begin{align*}
ALG = & p\cdot q\cdot E[x|x<tr]+(1-p)\cdot q \cdot E[x|x\geq tr]+p\cdot (1-q)\cdot E[y|y>tr]+(1-p)(1-q)\cdot E[x|x\geq tr]\\
=&  p\cdot q\cdot E[x|x<tr]+p\cdot (1-q)\cdot E[y|y>tr]+(1-p)\cdot E[x|x\geq tr]
\end{align*}

Since $E[y]=q\cdot E[y | y\leq tr] + (1-q)\cdot E[y | y> tr]$ we have that
$(1-q)\cdot E[y|y>tr]\geq E[y]-q\cdot t\cdot r$.  We continue in our effort to give a lower bound on ALG:
\begin{align*}
\geq &  p\cdot q\cdot E[x|x<tr]+p\cdot (E[y]-q\cdot t\cdot r)+(1-p)\cdot E[x|x\geq tr]\\
\geq & p\cdot E[y]-p\cdot q\cdot t\cdot r+(1-p)\cdot E[x|x\geq tr]
\end{align*}
Since $q\leq 1$ and since $ E[x|x\geq tr]\geq tr$:
\begin{align*}
\geq & p\cdot E[y]-p\cdot t\cdot r+(1-p)\cdot tr\\
\geq & t\cdot r+p\cdot (1-r)-2p\cdot t\cdot r
\end{align*}
Observe that whenever $t\leq \frac {1-r} {2r}$ we have that $2t\cdot r\leq  1-r$. Therefore, in this regime we the expression is minimized when $p$ gets the lowest value possible:
\begin{align*}
\geq & t\cdot r+(1-\frac 1 t)\cdot (1-r)-2(1-\frac 1 t)\cdot t\cdot r\\
=& 1+r-\frac 1 t+\frac r t-t\cdot r
\end{align*}
\end{proof}
To finish the case where $M_b\geq M_s$, consider first the case where $r< \frac 1 {13}$. In this case we choose $t=3$, and apply Lemma \ref{lemma-small-r}. The approximation ratio that we get is $\frac {13} 7$. On the other hand, if $r \geq \frac 1 {13}$, we apply Lemma \ref{lemma-mcafee-welfare} and get again an approximation ratio of $\frac {13} 7$.

\subsection*{Case II: $M_b<M_s$}

We show that there is a trade price that provides a $\frac {156} {255}$-approximation to the welfare for this case. We assume that $r>\frac 1 {13}$ -- otherwise Lemma \ref{lemma-small-r} guarantees a $\frac {13} 7$-approximation. For simplicity, in the proof of this case we assume that both distributions are discrete and that every atom in the support has the same probability. Both assumptions can be removed with some effort but the proof becomes a bit messier.

We start with some notation. For every element $x$ in the support of the seller's distribution $D_s$, let $p_x=\Pr_{q\sim D_s}[q<x]$. Consider an element $x<M_s$ in the support of $D_s$ (the distribution of the seller). We match $x$ to the (unique) element $x'$ in the support of $D_s$ with $p_{x'}=p_x+0.5$.
Using similar notation, we match every element $y<M_b$ in the support of $D_b$ to $y'$ in the support of $D_b$ with $p_{y'}=p_y+0.5$. Partition the set of all instances to 4-tuples (quadruples). We will have two types of quadruples:
\begin{enumerate}
\item Tuples where $M_b<y'<M_s$.
\item Tuples where $M_b<M_s<y'$.
\end{enumerate}

Fix some optimal mechanism. and denote by $O_1$ the contribution of type $1$ tuples to the optimum, and by $O_2$ the contribution of type $2$ tuples. Observe that $OPT=O_1+O_2$.

\begin{lemma}
The mechanism that posts a trade price $M_s$ provides an approximation ratio of $\frac {O_1} 2+\frac {7 O_2} {13}$.
\end{lemma}
\begin{proof}
We first show that we extract an expected welfare of at least $\frac {O_1} 2$ from type $1$ instances. In this case we have that $x'>M_s$ and that $y',y,x<M_s$. As usual we consider four instances: $<x,y>$, $<x,y'>$, $<x',y'>$, $<x',y>$. Since $x'$ is the largest value, the sum of the optimal solutions in all these four instances is trivially bounded from above by $4x'$. Now observe that in the last two instances the mechanism does not sell the item, so the sum of all the solutions that the mechanism outputs is bounded from below by $2x'$. This gives that we get a $\frac {O_1} 2$ from type $1$ instances.

Next we claim that we get a $2$ approximation for the gain from trade of type $2$ instances. Our constraints imply that the following quadruples as possible:
\begin{itemize}
\item $x'$ and $y'$ can take arbitrary values, as long as both are at least $M_s\geq M_b$.
\item $y\leq M_b\leq M_s$.
\item $x\leq M_s$.
\end{itemize}
It is routine to verify that all of these tuples are analyzed in Appendix \ref{sec-combinatorial-gft}. Thus we get a $2$ approximation for the gain from trade for all type $2$ tuples. Now, since $E[x]=E[x|\hbox{type }2]$ (because the distributions are independent, and whether a tuple is type $1$ or type $2$ is only a function of $y'$), we get that we extract at least $(\frac 1 2 +\frac r 2)\cdot O_2$ from type $2$ instances (as in Lemma \ref{lemma-mcafee-welfare}). Recall that by our assumption $r\geq \frac 1 {13}$, which finishes the proof of the lemma.
\end{proof}

\begin{lemma}
The mechanism that posts a trade price of $M_b$ provides an expected welfare of at least $\frac {2O_1} 3 $.
\end{lemma}
\begin{proof}
Next we consider type $1$ tuples when using this price. By definition of $x$,$y$,$x'$,$y'$ and type $1$ tuples, we have that $x'>M_s$, $M_b<y'<M_s$, $y<M_b$ and $x<M_s$. We consider two cases, depending on whether $x<M_b$ or not.

\begin{enumerate}
\item $x<M_b$. There are four possible instances:
\begin{enumerate}
\item In $<x,y>$ the optimal solution has value $\max(x,y)$ and we assume conservatively that the mechanism outputs a solution with value $\min(x,y)$.
\item In $<x',y>$ the optimal solution has value $x'$ is and the mechanism outputs a solution with value $x'$.
\item In $<x',y'>$ the optimal solution has value $x'$ and the mechanism outputs a solution with value $x'$.
\item In $<x,y'>$ the optimal solution has value $y'$ and the mechanism outputs a solution with value $y'$.
\end{enumerate}
One can check that the solution of the mechanism possibly differs from the solution of the optimal solution only in the first 
 case. Our loss in this case is bounded by $y<y'<x'$. Since the optimal aggregated welfare from all four instances is $2x'+y'+\max(x,y)$ and ALG's aggregated welfare is $2x'+y'+\min(x,y)$, we have that the expected approximation ratio is $\frac 4 3$ for the quadruple.

\item $M_s>x>M_b$. There are four possible instances:
\begin{enumerate}
\item In $<x,y>$ the optimal solution has value $x$ and the mechanism outputs solution with value $x$.
\item In $<x',y>$ the optimal solution has value $x'$ and the mechanism outputs solution with value $x'$.
\item In $<x',y'>$ the optimal solution has value $x'$ and the mechanism outputs solution with value $x'$
\item In $<x,y'>$ the optimal solution has value $\max(x,y')$ and we conservatively assume that the mechanism outputs solution with value $\min(x,y')$.
\end{enumerate}
\end{enumerate}
The solution of the mechanism possibly differs from the solution of the optimal solution only in the last case. In this case our loss is bounded by $x'$. Since the optimal aggregated welfare from all four instances is $2x'+x+\max(x,y')$ and ALG's aggregate welfare is $x+2x'+\min(x,y')$, we have that the expected approximation ratio is $\frac 3 2$ in this case.
\end{proof}

We therefore have two mechanisms that give two different guarantees: one provides an expected welfare of at least $\frac {O_1} {2}+\frac {7O_2} {13}$ and the other provides an expected welfare of at least $\frac {2O_1} {3} $. Recalling that $O_1+O_2=1$, we have that at least one of these mechanisms guarantees an approximation ratio of $\frac {28} {55}$.

\subsection{A Combinatorial Proof for McAfee's Mechanism}\label{sec-combinatorial-gft}


%
%
%
Let the seller's value be $x$ which is drawn from some distribution $D_s$ and the buyer's value be $y$, drawn from some independent distribution $D_b$. The expected \emph{gain from trade} is $E_{x\sim D_s,y\sim D_b}[\max(y-x,0)]$. McAffee shows the following:
\begin{theorem}[\cite{mcA08}]
Denote the median of $D_s$ by $M_s$ and the median of $D_b$ by $M_b$. If $M_b\geq M_s$ then setting a trade price of $M_s$ provides a $2$ approximation to the expected gain from trade.
\end{theorem}

We now provide an alternative combinatorial proof for McAffee's theorem. For simplicity, in the proof we assume that both distributions are discrete and that every atom in the support has the same probability. Both assumptions can be removed with some effort but the proof becomes a bit messier.

We start with some notation. For every element $x$ in the support of distribution $D_s$, let $p_x=\Pr_{q\sim D_s}[q<x]$. Consider an element $x<M_s$ in the support of $D_s$. We match $x$ to the (unique) element $x'$ in the support of $D_s$ with $p_{x'}=p_x+0.5$. Using similar notation, we match every element $y<M_b$ in the support of $D_b$ to $y'$ in the support of $D_b$ with $p_{y'}=p_y+0.5$.

We denote instances by $<x,y>$ where $x$ is the value of the seller and $y$ is the
value of the buyer. A \emph{losing pair} is an instance $<x,y>$ where $x<y$ but the
mechanism does not make a sell. If $x<y$ but the mechanism does make a sell this is a winning pair.

We can have the following three types of losing pairs, and we match each type with a winning pair. This way, all losing
pairs are matched to a winning pair.
\begin{enumerate}
\item \label{losing1} $<x,y>$ where $x<M_s$ and $y<M_b$. We match it to $<x,y'>$.
\item \label{losing2} $<x',y'>$ where $x'>M_s$ and $y'>M_b$. We match it to $<x,y'>$ as well.
\item \label{losing3} $<x',y>$ where $x',y$ are in $[M_s,M_b]$. If $x',y>p$, $<x',y>$ is matched to $<x,y>$. If $x',y < p$ then $<x',y>$ is matched to $<x',y'>$.
\end{enumerate}

We will show that the gain from trade from every winning pair is at least total gain from trade that would have been gained from the losing pairs matched to it.

Consider values $<x,y,x',y'>$ as in items \ref{losing1} and \ref{losing2}. $<x,y'>$ is a winning pair matched with two losing pairs $<x,y>$ and $<x',y'>$.
Since $y-x'<0$ (recall that $y<M_s$, $x'>M_b$), we have that $(y-x)+(y'-x')<y'-x$, and this what we wanted to show.

Consider values $<x,y,x',y'>$ as in item \ref{losing3}. Assume that the losing pair is $x',y>p$; Then, the winning pair is $<x,y>$ which is matched only to $<x',y>$, but $y-x\geq y-x'$. Similarly, assume that the losing pair is $x',y<p$; Then, the winning pair is $<x',y'>$ which is matched only to $<x',y>$, but $y'-x'\geq y-x'$.


Together, this gives us a $2$ approximation for the gain for trade.

\section{Missing Proofs Section \ref{sec-multiunit}}

\subsection*{Sketch of Proof of Lemma \ref{lemma-multiunit-no-substantial}}

We divide into several cases:
\begin{itemize}
\item \emph{There are two bidders $i$ and $j$ with $r_i,r_j\geq \frac 1 8$}: we select at random one player $i$ or $j$ to play $N_2$ and the remaining one to play $N_3$. The rest of the players will consist the group $N_1$. Observe that if either $v_i(o_i)\geq \frac {OPT} 6 $ or $v_j(o_j)\geq \frac {OPT} 6 $ then we are already done, since the current value for the endowment is at least $\frac 1 8$ of their contribution to the optimal welfare. Furthermore, no individually rational mechanism (and our mechanism in particular) can decrease the welfare.

We may therefore assume that the total contribution of bidders in $N_1$ to the optimal welfare is at least $\frac {2OPT} 3$. Running the mechanism (but now with $t=\frac 1 {16}$) and using the very similar analysis, give us welfare of at least $\frac {2OPT} 3\cdot \frac 1 {16}$ with probability $\frac 1 2$, hence the approximation ratio is $48$ in this case.

\item \emph{There is at most one bidder $i$ with $r_i\geq \frac 1 8$}: we show how to construct three substantial sets and therefore we can run the mechanism with no change.

Recall that by assumption $r_i\leq \frac 1 3$. Select two disjoint minimal sets of bidders $T_1$ and $T_2$ such that for each $l$, $\Sigma_{i'\in T_l}r_{i'}\geq \frac 1 4$ and $i\notin T_l$. Notice that by minimality and since for each bidder $i'\neq i$ we have that $r_{i'}\leq \frac 1 8$, we have that $\Sigma_{i'\in T_l}r_{i'}\leq \frac 3 8$. Notice that two such minimal sets exist since $\Sigma_{i'\neq i}r_{i'}\geq \frac 2 3$.

Let our three substantial sets be $T_1, T_2, N-T_1-T_2$. $T_1$ and $T_2$ are substantial by construction, and that $N-T_1-T_2$ is substantial since $\Sigma_{i'\in T_l}r_{i'}\leq \frac 3 8$ for each $l$.
\end{itemize}

\end{document}